\documentclass{amsart}
\usepackage{tikz,amsmath,amssymb,color,enumerate}

\theoremstyle{plain}
\newtheorem{theorem}{Theorem}
\newtheorem{lemma}[theorem]{Lemma}
\newtheorem{proposition}[theorem]{Proposition}
\newtheorem{corollary}[theorem]{Corollary}
\theoremstyle{definition}
\newtheorem{definition}[theorem]{Definition}

\newtheorem{remark}[theorem]{Remark}
\newtheorem{question}[theorem]{Question}

\DeclareMathOperator{\colim}{colim}
\renewcommand{\Im}{\ensuremath{\mathrm{Im}}}
\newcommand{\cat}[1]{\ensuremath{\mathbf{#1}}}
\DeclareMathOperator{\Sub}{Sub}
\DeclareMathOperator{\Max}{Max}
\DeclareMathOperator{\At}{At}
\DeclareMathOperator{\PBoolD}{Spec}
\DeclareMathOperator{\id}{id}
\DeclareMathOperator{\Id}{Id}
\DeclareMathOperator{\height}{ht}

\newcommand{\op}{\ensuremath{{}^{\text{op}}}}
\renewcommand{\restriction}{\mathord{\upharpoonright}}

\usetikzlibrary{arrows}
\newcommand*\arrowoffset{0.4pt}
\makeatletter
\pgfarrowsdeclare{new double arrowhead}{new double arrowhead}
{
  \pgfutil@tempdima=-0.84pt%
  \advance\pgfutil@tempdima by-1.3\pgflinewidth%
  \pgfutil@tempdimb=0.21pt%
  \advance\pgfutil@tempdimb by.625\pgflinewidth%
  \pgfarrowsleftextend{+\pgfutil@tempdima}
  \advance\pgfutil@tempdimb by \arrowoffset%
  \pgfarrowsrightextend{+\pgfutil@tempdimb}
}
{
  \pgfsetstrokecolor{black}
  \pgfsetlinewidth{0.6pt}
  \pgfutil@tempdima=0.28pt%
  \advance\pgfutil@tempdima by.3\pgflinewidth%
  \pgfsetlinewidth{0.8\pgflinewidth}
  \pgfsetdash{}{+0pt}
  \pgfsetroundcap
  \pgfsetroundjoin
  \pgfpathmoveto{\pgfpoint{-3\pgfutil@tempdima+\arrowoffset}{4\pgfutil@tempdima}}
  \pgfpathcurveto
{\pgfpoint{-2.75\pgfutil@tempdima+\arrowoffset}{2.5\pgfutil@tempdima}}
                {\pgfpoint{\arrowoffset}{0.25\pgfutil@tempdima}}
                {\pgfpoint{0.75\pgfutil@tempdima+\arrowoffset}{0pt}}
  \pgfpathcurveto
                {\pgfpoint{\arrowoffset}{-0.25\pgfutil@tempdima}}
{\pgfpoint{-2.75\pgfutil@tempdima+\arrowoffset}{-2.5\pgfutil@tempdima}}
{\pgfpoint{-3\pgfutil@tempdima+\arrowoffset}{-4\pgfutil@tempdima}}
  \pgfusepathqstroke
}
\makeatother

\begin{document}
\title{Piecewise Boolean algebras and their domains}
\author{Chris Heunen}
\address{University of Oxford, Department of Computer Science}
\email{heunen@cs.ox.ac.uk}
\date{\today}
\maketitle
\begin{abstract}
  We characterise piecewise Boolean domains, that is, those domains that arise as Boolean subalgebras of a piecewise Boolean algebra.
  This leads to equivalent descriptions of the category of piecewise Boolean algebras: either as piecewise Boolean domains equipped with an orientation, or as full structure sheaves on piecewise Boolean domains. 
\end{abstract}

\section{Introduction}

Boolean algebras embody the logical calculus of observations. 
But in many applications it does not make sense to consider any two observations simultaneously.
For a simple example, can you really verify that ``there is a polar bear in the Arctic'' and ``there is a penguin in Antarctica'', when you cannot be in both places at once?
This leads to the notion of a \emph{piecewise Boolean algebra}\footnote{N{\'e}e \emph{partial} Boolean algebra; recent authors use \emph{piecewise} to avoid `partial complete Boolean algebra'~\cite{vdbergheunen:colim}. Incidentally, this is the structure Boole originally studied~\cite{hailperin:boole}.}, which is roughly a Boolean algebra where only certain pairs of elements have a conjunction.

You could say that the issue in the above example is merely caused by a constructive interpretation. But it is a real, practical concern in \emph{quantum logic}, where the laws of nature forbid jointly observing certain pairs (the famous example being to measure position and momentum),
and piecewise Boolean algebras consequently play a starring role~\cite{hughes:omnibus,hughes:semantic,
finch:structure,gudder:partial,kalmbach:orthomodularlattices}. 

Another cause of incompatible observations relates to \emph{partiality}. Some (observations of) computations might not yet have returned a result, but nevertheless already give some partial information. It might not make sense to compare two partial observations, whereas the completed observations would be perfectly compatible. Partiality is also at play in quantum theory, where measurements can be fine-grained, so that the course-grained version only gives partial information.
This leads to \emph{domain theory}~\cite{abramskyjung:domaintheory,jung:domains}.

This paper brings the two topics, domain theory and quantum logic, together. The main construction sends a piecewise Boolean algebra $P$ to the collection $\Sub(P)$ of its compatible parts, \textit{i.e.}\ its Boolean subalgebras. This well-known construction~\cite{vdbergheunen:colim,hughes:omnibus,finch:structure,gudder:partial,heunenlandsmanspitters:topos,doeringbarbosa:domains,hardingnavara:subalgebras,graetzerkohmakkai:booleansubalgebras,sachs:booleansubalgebras} assigns a domain $\Sub(P)$ to a piecewise Boolean algebrac $P$. Our main result is a characterisation of the domains of the form $\Sub(P)$, called \emph{piecewise Boolean domains}; it turns out they are the so-called algebraic L-domains whose bottom two rungs satisfy some extra properties. 
This gives an alternative description of piecewise Boolean algebras, that is more concise, amenable to domain theoretic techniques, and addresses open questions~\cite[Problems~1 and~2]{hardingnavara:subalgebras}. 
Colloquially, it shows that to reconstruct the whole, it suffices to know how the parts fit together, without having to know the internal structure of the parts.

Commutative rings, such as Boolean algebras, can be reconstructed from their Zariski spectrum together with the structure sheaf over that spectrum~\cite[V.3]{johnstone:stonespaces}. Analogously, we prove that a piecewise Boolean algebra can be reconstructed from its piecewise Boolean domain together with the structure sheaf over that domain. (Equivalently, we could use the Stone dual of the structure sheaf.) We prove a categorical equivalence between piecewise Boolean algebras, and piecewise Boolean domains with a subobject-preserving functor valued in Boolean algebras. We call the latter objects \emph{piecewise Boolean diagrams}. 

There is a beautiful microcosm principle at play in the reconstruction of a piecewise Boolean diagram from a piecewise Boolean domain: piecewise Boolean diagrams are really structure-preserving functors from a piecewise Boolean domain into the category of Boolean algebras. The piecewise Boolean diagram is almost completely determined by the piecewise Boolean domain, but some choices have to be made. We condense those choices into an \emph{orientation}, that fixes a choice between two possibilities on each atom of a piecewise Boolean domain. Finally, we prove that the category of piecewise Boolean algebras is equivalent to the category of oriented piecewise Boolean domains.

We proceed as follows. Section~\ref{sec:pbool} recalls the basics of piecewise Boolean algebras, after which Section~\ref{sec:spectralposets} introduces piecewise Boolean domains and proves they are precisely those domains of the form $\Sub(P)$. This characterisation is simplified further in Section~\ref{sec:partitionlattices}. Section~\ref{sec:spectraldiagrams} proves the equivalence between piecewise Boolean algebras and piecewise Boolean diagrams, and Section~\ref{sec:orientation} reduces from piecewise Boolean diagrams to oriented piecewise Boolean domains. Finally, Section~\ref{sec:future} concludes with directions for future work. For example, it would be interesting to explore connections to other work~\cite{laird:locallyboolean,abramskyvickers:quantales,kozenetal:stonemarkov}.

\section{Piecewise Boolean algebras}\label{sec:pbool}

\begin{definition}
  A \emph{piecewise Boolean algebra} consists of a set $P$ with:
  \begin{itemize}
    \item a reflexive and symmetric binary (\emph{commeasurability}) relation $\odot \subseteq P \times P$;
    \item elements $0,1 \in P$;
    \item a (total) unary operation $\neg \colon P \to P$;
    \item (partial) binary operations $\wedge,\vee \colon \odot \to P$;
  \end{itemize}
  such that every set $A \subseteq P$ of pairwise commeasurable elements is contained in a set $B \subseteq P$, whose elements are also pairwise commeasurable, and on which the above operations determine a Boolean algebra structure.

  A morphism of piecewise Boolean algebras is a function that preserves commeasurability and all the algebraic structure, whenever defined. Piecewise Boolean algebras and their morphisms form a category $\cat{PBool}$.
\end{definition}

A piecewise Boolean algebra in which every two elements are commeasurable is just a Boolean algebra. Given a piecewise Boolean algebra $P$, we write $\Sub(P)$ for the collection of its commeasurable subalgebras, ordered by inclusion. (The maximal elements of $\Sub(P)$ are also called \emph{blocks}, see~\cite[Section~1.4]{kalmbach:orthomodularlattices}.)
In fact, $\Sub$ is a functor $\cat{PBool} \to \cat{Poset}$ to the category of partially ordered sets and monotone functions, acting on morphisms by direct image. If $P$ is a piecewise Boolean algebra, $\Sub(P)$ is called its \emph{piecewise Boolean domain}.

We now list two main results about piecewise Boolean algebras and their domains. 
First, we can reconstruct $P$ from $\Sub(P)$ up to isomorphism.

\begin{theorem}[\cite{vdbergheunen:colim}]\label{thm:colim}
  Any piecewise Boolean algebra $P$ is a colimit of $\Sub(P)$.
\end{theorem}

Boolean algebras are precisely objects of the ind-completion of the category of finite Boolean algebras~\cite[VI.2.3]{johnstone:stonespaces}, defining Boolean algebras as colimits of diagrams of finite Boolean algebras. The previous theorem extends this to piecewise Boolean algebras.
Second, $\Sub(P)$ determines $P$ up to isomorphism.

\begin{theorem}[\cite{hardingnavara:subalgebras}]\label{thm:hardingnavara}
  If $P$ and $P'$ are piecewise Boolean algebras and $\varphi \colon \Sub(P) \to \Sub(P')$ is an isomorphism, then there is an isomorphism $f \colon P \to P'$ with $\varphi = \Sub(f)$.
  Moreover, $f$ is unique iff atoms of $\Sub(P)$ are not maximal.
\end{theorem}

However, the functor $\Sub$ is not an equivalence. It is not faithful: see the above theorem. Neither is it full: not every monotone function $\Sub(P) \to \Sub(P')$ preserves atoms.
Nevertheless, the previous two theorems show that the functor $\Sub$ is almost an equivalence.
Later, we will upgrade the functor $\Sub$ to an equivalence. But first we investigate posets of the form $\Sub(P)$.

\section{Piecewise Boolean domains}\label{sec:spectralposets}

This section characterises piecewise Boolean domains in terms of finite partition lattices, which we will characterise further in the next section.
Recall that an element $x$ of a poset $P$ is \emph{compact} when, if $x \leq \bigvee D$ for a directed subset $D \subseteq P$ with a supremum, then $x \leq y$ for some $y \in D$.
Write $K(P)$ for the partially ordered set of compact elements of $P$.

\begin{definition}\label{def:spectralposet}
  A poset is called a \emph{piecewise Boolean domain} when:
  \begin{enumerate}
    \item[(1)] it has directed suprema;
    \item[(2)] it has nonempty infima;
    \item[(3)] each element is the directed supremum of compact ones;
    \item[(4)] the downset of each compact element is dual to a finite partition lattice.
  \end{enumerate}
\end{definition}

Posets satisfying properties (1)--(3) are also known as \emph{Scott domains}~\cite{scott:denotational}.

\begin{proposition}
  If $P$ is a piecewise Boolean algebra, $\Sub(P)$ is a piecewise Boolean domain.
\end{proposition}
\begin{proof}
  If $B_i \in \Sub(P)$, then also $\bigcap B_i \in \Sub(P)$, giving nonempty infima.
  If $\{B_i\}$ is a directed family of elements of $\Sub(P)$, then $\bigcup B_i$ is a Boolean algebra, which is the supremum in $\Sub(P)$. 
  To show that every element is the directed supremum of compact ones, it therefore suffices to show that the compact elements are the finite Boolean subalgebras of $P$. But this is easily verified.
  Finally, the downset of any compact element is pairwise commeasurable, hence a finite Boolean algebra, and it is dual to a finite partition lattice.~\cite{graetzerkohmakkai:booleansubalgebras,sachs:booleansubalgebras}.
\end{proof}

We now set out to prove that any piecewise Boolean domain $L$ is of the form $\Sub(P)$ for some piecewise Boolean algebra $P$. The first step is to show $L$ gives rise to a functor $L \to \cat{Bool}$ that preserves the structure of $L$. 
For $x \in L$, we write $\Sub(x)$ for the principal ideal of $x$. 

\begin{remark}
  Both occurrences of $\Sub$ are instances of a more general scheme. If $\cat{C}$ is a category with epi-mono factorizations, we write $\Sub \colon \cat{C} \to \cat{Poset}$ for the covariant subobject functor. It acts as direct image on morphisms $f \colon x \to y$, that is, a subobject $m \colon \bullet \rightarrowtail x$ gets mapped to the image $f[m] \colon \Im(f \circ m) \rightarrowtail y$. If $\cat{C}$ is a poset, then $\Sub(x)$ is just the principal ideal of $x$, and functoriality just means that $\Sub(x) \subseteq \Sub(y)$ when $x \leq y$. If $\cat{C}=\cat{Bool}$, then $\Sub(B)$ is the lattice of Boolean subalgebras of $B$, and the direct image $f[A]$ of a Boolean subalgebra $A$ under a homomorphism $f \colon B \to B'$ is a Boolean subalgebra of $B'$. By slight abuse of notation, if $\cat{C}$ is the category $\cat{PBool}$, we let $\Sub(P)$ be the poset of Boolean subalgebras of $P$ (instead of piecewise Boolean subalgebras), as before. The action on morphisms by direct image is then still well-defined. 
\end{remark}



\begin{lemma}\label{lem:functor}
  Let $L$ be a piecewise Boolean domain.
  \begin{enumerate}[(a)]
    \item For each $x \in L$ there is a Boolean algebra $F(x)$ with $\Sub(F(x)) \cong \Sub(x)$.
    \item There is a functor $F \colon L \to \cat{Bool}$ and a natural isomorphism $\Sub \circ F \cong \Sub$. 
  \end{enumerate}
\end{lemma}
\begin{proof}
  Properties (1) and (2) make $L$ into an L-domain~\cite[Theorem~2.9]{jung:domains}.
  Adding property (3) makes $L$ into an algebraic L-domain~\cite[Section~2.2]{jung:domains}.
  It follows that every downset is an algebraic lattice~\cite[Corollary~1.7 and Proposition~2.8]{jung:domains}, and in fact that $\bigcup_x K(\Sub(x)) = K(L)$~\cite[Proposition~1.6]{jung:domains}. Finally, property (4) ensures that every downset satisfies the following property: it is an algebraic lattice, and each compact element in it is dual to a finite partition lattice. Therefore every downset is the lattice of Boolean subalgebras of some Boolean algebra~\cite{graetzerkohmakkai:booleansubalgebras}, establishing (a).

  Towards (b), define $\varphi_{x,y}$ for $x \leq y \in L$ as the following composition.
  \[\begin{tikzpicture}[xscale=2.5, yscale=1.33]
    \node (tl) at (0,1) {$\Sub(x)$};
    \node (bl) at (0,0) {$\Sub(y)$};
    \node (tr) at (1,1) {$\Sub(F(x))$};
    \node (br) at (1,0) {$\Sub(F(y))$};
    \draw[->] (tr) to node[above] {$\cong$} (tl);
    \draw[->] (bl) to node[below] {$\cong$} (br);
    \draw[->] (tl) to node[left] {$\Sub(x \leq y)$} (bl);
    \draw[->] (tr) to node[right] {$\varphi_{x,y}$} (br);
  \end{tikzpicture}\]
  Because $\Sub(x \leq y)$ is a monomorphism of complete lattices~\cite[Proposition~2.8]{jung:domains}, so is $\varphi_{x,y}$. Now, $\Sub(\varphi_{x,y})(\Sub(F(x))) \in \Sub(\Sub(F(y)))$; that is, the direct image of $\varphi_{x,y}$ is downward closed in $\Sub(F(y))$. So, by construction, the direct image of $\varphi_{x,y}$ is $\Sub(B)$, where $B=\varphi_{x,y}(F(x))$. Hence $\varphi_{x,y}$ factors as an isomorphism $\psi \colon \Sub(F(x)) \to \Sub(B)$ followed by an inclusion $\Sub(B) \subseteq \Sub(F(y))$. By~\cite[Theorem~4]{graetzerkohmakkai:booleansubalgebras} or~\cite[Corollary~2]{sachs:booleansubalgebras}, there is an isomorphism $f \colon F(x) \to B$ such that $\psi=\Sub(f)$. Also, $B \in \Sub(B) \subseteq \Sub(F(y))$, so $B$ is a Boolean subalgebra of $F(y)$. That is, there is an inclusion $g \colon B \hookrightarrow F(y)$ such that $\Sub(g)$ is the inclusion $\Sub(B) \subseteq \Sub(F(y))$.
  Thus $F(x \leq y) := g \circ f \colon F(x) \rightarrowtail F(y)$ is a monomorphism of Boolean algebras that satisfies $\Sub(F(x \leq y)) = \varphi_{x,y}$.
  If $|F(x)| \neq 4$, then $F(x \leq y)$ is in fact the unique such map~\cite[Lemma~5]{graetzerkohmakkai:booleansubalgebras}, and in this case it follows that $F(y \leq z) \circ F(x \leq y) = F(x \leq z)$.

  Next, we will adjust $F(x \leq y)$ for $|F(x)|=4$ if need be, to ensure functoriality of $F$.
  Let $x$ be an atom of $L$. If $x$ is maximal, there is nothing to do. Otherwise choose $y$ covering $x$. Select one of the two possible $F(x<y)$ inducing $\varphi_{x,y}$. Now, for any $y'>x$ such that $z=y \vee y'$ exists we need to choose $F(x<y')$ making the following diagram commute.
  \begin{equation}\label{eq:orientationchoice}\tag{$*$}
  \begin{aligned}\begin{tikzpicture}[xscale=3,yscale=1.33]
    \node (tl) at (0,1) {$F(x)$};
    \node (bl) at (0,0) {$F(y')$};
    \node (tr) at (1,1) {$F(y)$};
    \node (br) at (1,0) {$F(z)$};
    \draw[->] (tl) to node[above] {$F(x<y)$} (tr);
    \draw[->] (tl) to node[left] {$F(x<y')$} (bl);
    \draw[->] (tr) to node[right] {$F(y<z)$} (br);
    \draw[->] (bl) to node[below] {$F(y'<z)$} (br);
  \end{tikzpicture}\end{aligned}\end{equation}
  Let us write $\alpha_z$ for the isomorphism $\Sub(F(z)) \to \Sub(z)$. Next, notice that $X:=F(y<z) \circ F(x<y)[F(x)] = \varphi_{x,z}(F(x)) = \alpha_z(x) \subseteq F(z)$, and similarly $Y:=F(y'<z)[F(y')]=\varphi_{y',z}(F(y'))=\alpha_z(y') \subseteq F(z)$; because $x<y'$ hence $X \subseteq Y$, and there is a unique $F(x<y')$ making the diagram commute. Moreover $\Sub(F(x<y'))=\varphi_{x,y'}$.
  Thus $F$ is functorial, and the isomorphisms $\Sub \circ F \cong \Sub$ are natural by construction. This proves part (b).
%
\end{proof}

We say a functor $F \colon L \to \cat{Bool}$ \emph{preserves subobjects} when there is a natural isomorphism $\Sub \circ F \cong \Sub$.

Next, we show that the data contained in the functor $L \to \cat{Bool}$ can equivalently be packaged as a piecewise Boolean algebra by taking its colimit.

\begin{lemma}\label{lem:colimit}
  Let $L$ be a piecewise Boolean domain, let $F$ be the functor of Lemma~\ref{lem:functor}, and 
  let the piecewise Boolean algebra $P$ be the colimit of $F$ in $\cat{PBool}$.
  \begin{enumerate}[(a)]
    \item Maximal elements of $L$ correspond bijectively to maximal elements of $\Sub(P)$.
    \item The colimit maps $F(x) \to P$ are injective.
  \end{enumerate}
\end{lemma}
\begin{proof}
  In general, colimits of piecewise Boolean algebras are hard to compute (see~\cite[Theorem~2]{vdbergheunen:colim}, and also~\cite{haimo:limits}). But injectivity of $F(x \leq y)$ makes it  managable. Namely, $P = \coprod_{x \in L} F(x) / \sim$, where $\sim$ is the smallest equivalence relation satisfying $b \sim F(x \leq y)(b)$ when $x \leq y$ and $b \in F(x)$.
  That is, $F(x_1) \ni b_1 \sim b_n \in F(x_n)$ means there are $x_2,\ldots,x_{n-1} \in L$ with $x_1 \geq x_2 \leq x_3 \geq x_4 \leq x_5 \geq \cdots \geq x_{n-1} \leq x_n$, and $b_i \in F(x_i)$ for $i=2,\ldots,n-1$ that satisfy $b_{i+1}=F(x_i \leq x_{i+1})(b_i)$ for even $i$ and $b_i = F(x_{i+1} \leq x_i)(x_{i+1})$ for odd $i$. Let us write $p_x \colon F(x) \to P$ for the colimiting maps $p_x(a) = [a]_\sim$.

  If $x_1$ and $x_n$ are maximal, then without loss of generality we may assume that $x_i$ is maximal for odd $i$ and that $x_{i+1} = x_i \wedge x_{i+2}$ for odd $i$. By the naturality of Lemma~\ref{lem:functor}(b), this means that the subalgebra $F(x_2)$ of $F(x_1)$ and $F(x_3)$ is identified. So, by injectivity of $F(x \leq y)$, the only way the entire algebra $F(x_1)$ can be identified with $F(x_n)$ is when $x_1=\ldots=x_n$. 

  Define a function $f \colon \Max(L) \to \Max(\Sub(P))$ by $f(x) = p_x[F(x)] = [F(x)]_\sim$.
  The discussion above shows that $f$ is injective.
  Any $B \in \Sub(P)$ is commeasurable, and hence there is $x \in L$ such that $B \subseteq [F(x)]_\sim$. If $B$ is maximal, then we must have $B=f(x)$. Thus $f$ is well-defined, and surjective. This proves (a).

  For part (b), let $x \in L$. It follows from Zorn's Lemma and property~(1) that $x$ is below some maximal $y \in L$. By part (a), then $p_y$ is injective. Therefore $p_x = p_y \circ F(x \leq y)$ is injective, too.
\end{proof}

We are now ready to prove our main result.

\begin{theorem}\label{thm:characterisation}
  Any piecewise Boolean domain is isomorphic to $\Sub(P)$ for a piecewise Boolean algebra $P$.
\end{theorem}
\begin{proof}
  Let $L$ be a piecewise Boolean domain. Fix a functor $F$ as in Lemma~\ref{lem:functor}, and its piecewise Boolean algebra colimit $p_x \colon F(x) \to P$ as in Lemma~\ref{lem:colimit}. Define $f \colon L \to \Sub(P)$ as $f(x) = p_x[F(x)]$.

  We first prove that $f$ is surjective. Any $B \in \Sub(P)$ is commeasurable, so $B$ is a Boolean subalgebra of $p_y[F(y)]$ for some $y \in L$. Hence $p_y^{-1}(B) \in \Sub(F(y))$. Because $F$ preserves subobjects, $p_y^{-1}(B) = F(x \leq y)[F(x)]$ for some $y \leq x$. Then:
  \[
    f(x) 
    = p_x[F(x)]
    = p_y \circ F(x \leq y)[F(x)]
    = p_y [p_y^{-1}(B)]
    = B.
  \]

  Next we prove that $f$ is injective by exhibiting a left-inverse $g \colon \Sub(P) \to L$. Set $g(B) = \bigwedge \{ x \in L \mid B \subseteq f(x) \}$. Note that $g(f(x)) = \bigwedge \{ y \mid [F(x)]_\sim \subseteq [F(y)]_\sim \} \leq x$. Now, if $y \leq x$ then $[F(y)]_\sim = p_y[F(y)] = p_x \circ F(y \leq x)[F(y)] \subseteq [F(x)]_\sim$. Hence if also $[F(x)]_\sim \subseteq [F(y)]_\sim$, then $F(y \leq x)$ is an isomorphism, and $x=y$. So $g(f(x))=x$.

  Clearly $g(B) \leq g(C)$ when $B \subseteq C$, so $f(x) \subseteq f(y)$ implies $x \leq y$.
  Conversely, if $x \leq y$, then $f(x) = p_x[F(x)] = p_y[F(x \leq y)[F(x)]] \subseteq p_y[F(y)] = f(y)$.
  Thus $f$ is an order isomorphism $\Sub(P) \cong L$.
\end{proof}


\section{Partition lattices}\label{sec:partitionlattices}

There exist many characterisations of finite partition lattices~\cite{sachs:partitionlattices,sasakifujiwara:partitionlattices,ore:equivalencerelations,firby:compactifications1,aigner:uniformitaet,stonesiferbogart:partition,yoon:partition}.
We now summarise one of them that we will use to reformulate condition~(4).
In a partition lattice, the intervals $[p,1]$ for atoms $p$ are again partition lattices. This leads to the following result. 
For terminology, recall that a finite lattice is \emph{(upper) semimodular}
when $x$ covers $x \wedge y$ implies that $x \vee y$ covers $y$, 
that a \emph{geometric lattice} is a finite atomistic semimodular lattice, and that an element $x$ of a lattice is called \emph{modular} if $a \vee (x \wedge y) = (a \vee x) \wedge y$ for all $a \leq y$.

\begin{theorem}[\cite{stonesiferbogart:partition,yoon:partition}]
\label{thm:stonesiferbogart}
  Suppose $L$ is a geometric lattice with a modular coatom, and the interval $[p,1]$ is a partition lattice of height $n-1$ for all atoms $p$.
  If $n\leq 4$, assume that $L$ has $n \choose 2$ atoms.
  Then $L$ is a partition lattice of height $n$. 
  Conversely, a partition lattice of height $n$ satisfies these requirements.
\end{theorem}


Let us call a lattice \emph{cogeometric} when it is dual to a geometric lattice; this is equivalent to being finite, lower semimodular, and coatomistic. 
We can now simplify condition~(4), showing that piecewise Boolean domains are domains that are determined entirely by their behaviour on the bottom three rungs.

\begin{proposition}\label{prop:characterisation}
  A poset is a piecewise Boolean domain precisely if it meets conditions~(1)--(3) and
  \begin{enumerate}
    \item[(4')] the downset of a compact element is cogeometric and has a modular atom;
    \item[(4'')] each element of height $n \leq 3$ covers exactly $n+1 \choose 2$ elements.
  \end{enumerate}
\end{proposition}
\begin{proof}
  We show that we may replace condition~(4) in Definition~\ref{def:spectralposet} by (4') and (4'').
  Observe that a dual lattice having a modular coatom is equivalent to the lattice itself having a modular atom. 
  Assuming condition~(4) and $x \in K(L)$, then $\Sub(x)$ is dual to a finite partition lattice, so that condition~(4') is satisfied. For $\height(x) \leq 4$, condition~(4'') is  verified by computing the partition lattices of height up to three, see Figure~\ref{fig:smallpartitionlattices}.

  Conversely, assume~(4') and~(4''). Then the downset of each compact element is finite, so that compact elements have finite height. Hence condition~(4) follows by induction on the height by Theorem~\ref{thm:stonesiferbogart}.
\end{proof}

\begin{figure}[h]
  \[
  \Pi_1 = \begin{aligned}\begin{tikzpicture}[scale=0.75,inner sep=0pt,font=\tiny]
      \node at (0,0) {$1$}; 
    \end{tikzpicture}\end{aligned}
  \qquad\qquad
  \Pi_2 = \begin{aligned}\begin{tikzpicture}[scale=0.75,inner sep=0pt,font=\tiny,inner sep=1pt]
      \node (l) at (0,0) {$12$}; \node (t) at (0,1) {$1/2$}; \draw (l) to (t); 
    \end{tikzpicture}\end{aligned}
  \qquad\qquad
  \Pi_3 = \begin{aligned}\begin{tikzpicture}[yscale=0.8,inner sep=0pt,font=\tiny,inner sep=1pt]
      \node (l) at (0,.33) {$1/2/3$};
      \node (a) at (-1,1) {$1/23$};
      \node (b) at (0,1) {$13/2$};
      \node (c) at (1,1) {$12/3$};
      \node (t) at (0,1.66) {$123$};
      \draw (l) to (a); \draw (l) to (b); \draw (l) to (c);
      \draw (a) to (t); \draw (b) to (t); \draw (c) to (t);
    \end{tikzpicture}\end{aligned}
  \]
  \[\Pi_4 = \begin{aligned}\begin{tikzpicture}[xscale=1.66,yscale=1.3,font=\tiny,inner sep=1pt]
    \node (l) at (0,.33) {$1/2/3/4$};
    \node (l1) at (-2.5,1) {$12/3/4$};
    \node (l2) at (-1.5,1) {$13/2/4$};
    \node (l3) at (-.5,1) {$14/2/3$};
    \node (l4) at (.5,1) {$1/23/4$};
    \node (l5) at (1.5,1) {$1/3/24$};
    \node (l6) at (2.5,1) {$1/2/34$};
    \node (t1) at (-3,2) {$123/4$};
    \node (t2) at (-2,2) {$124/3$};
    \node (t3) at (-1,2) {$13/24$};
    \node (t4) at (0,2) {$12/34$};
    \node (t5) at (1,2) {$14/23$};
    \node (t6) at (2,2) {$134/2$};
    \node (t7) at (3,2) {$1/234$};
    \node (t) at (0,2.66) {$1234$};
    \draw (l.north) to (l1); \draw (l.north) to (l2); \draw (l.north) to (l3);
    \draw (l.north) to (l4); \draw (l.north) to (l5); \draw (l.north) to (l6);
    \draw (t.south) to (t1); \draw (t.south) to (t2); \draw (t.south) to (t3);
    \draw (t.south) to (t4); \draw (t.south) to (t5); \draw (t.south) to (t6); \draw (t.south) to (t7);
    \draw (l1.north) to (t1.south); \draw (l1.north) to (t2.south); \draw (l1.north) to (t4.south);
    \draw (l2.north) to (t1.south); \draw (l2.north) to (t3.south); \draw (l2.north) to (t6.south);
    \draw (l3.north) to (t2.south); \draw (l3.north) to (t5.south); \draw (l3.north) to (t6.south);
    \draw (l4.north) to (t1.south); \draw (l4.north) to (t5.south); \draw (l4.north) to (t7.south);
    \draw (l5.north) to (t2.south); \draw (l5.north) to (t3.south); \draw (l5.north) to (t7.south);
    \draw (l6.north) to (t4.south); \draw (l6.north) to (t6.south); \draw (l6.north) to (t7.south);
  \end{tikzpicture}\end{aligned}
  \]
  \caption{The partition lattices of height up to three.}
  \label{fig:smallpartitionlattices}
\end{figure}
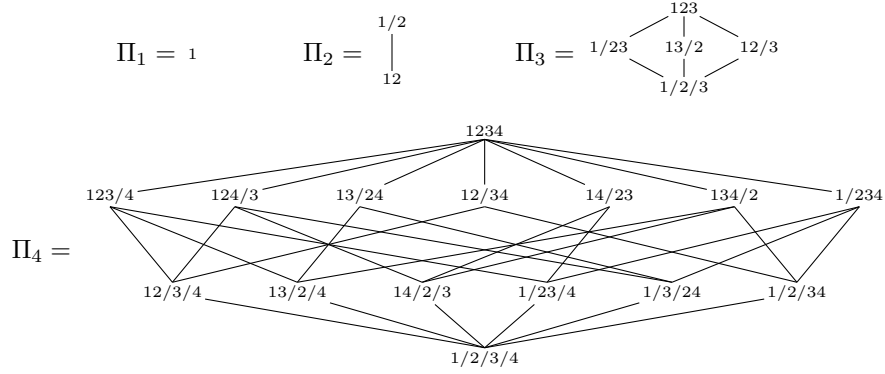

\section{Piecewise Boolean diagrams}\label{sec:spectraldiagrams}

\begin{definition}\label{def:spectraldiagram}
  A \emph{piecewise Boolean diagram} is a subobject-preserving functor from a piecewise Boolean domain to $\cat{Bool}$. 
  A morphism of piecewise Boolean diagrams from $F \colon L \to \cat{Bool}$ to $F' \colon L' \to \cat{Bool}$ consists of a morphism $\varphi \colon L \to L'$ of posets and a natural transformation $\eta \colon F \Rightarrow F' \circ \varphi$.
  Piecewise Boolean diagrams and their morphisms form a category $\cat{PBoolD}$. Composition is given by $(\psi,\theta)\circ(\varphi,\eta) = (\psi \circ \varphi, \theta\varphi \cdot \eta)$, and identies are $(\id,\Id)$.
  \[\begin{tikzpicture}[xscale=1.5,yscale=1.3]
    \node (l) at (-1,1) {$L$};
    \node (t) at (0,1) {$L'$};
    \node (b) at (0,0) {$L''$};
    \node (r) at (2,1) {$\cat{Bool}$};
    \draw[->] (l) to[out=45, in=135] node[above] {$F$} (r);
    \draw[->] (t) to node[above] {$F'$} (r);
    \draw[->] (b) to node[below] {$F''$} (r);
    \draw[->] (t) to node[left] {$\psi$} (b);
    \draw[->] (l) to node[above] {$\varphi$} (t);
    \draw[-new double arrowhead,double,double equal sign distance] (0.5,.9) to node[right] {$\theta$} (0.5,0.33);
    \draw[-new double arrowhead,double,double equal sign distance] (0.5,1.65) to node[right] {$\eta$} (0.5,1.1);
  \end{tikzpicture}\]
\end{definition}

Notice that, because $F$ preserves subobjects, also $\Sub(\varphi(x))=\varphi[\Sub(x)]$, so that $\varphi$ preserves directed suprema.

The functor $\Sub$ extends from piecewise Boolean domains to piecewise Boolean diagrams as follows. 

\begin{proposition}\label{prop:specfunctor}
  There is a functor $\PBoolD \colon \cat{PBool} \to \cat{PBoolD}$ defined as follows. On objects $P \in \cat{PBool}$, define $\PBoolD(P) \colon \Sub(P) \to \cat{Bool}$ by $B \mapsto B$. On morphisms $f \colon P \to P'$, define $\PBoolD(f)_B = f\restriction_B \colon B \to f[B]$.
  \qed
\end{proposition}

There is also a functor in the other direction. We will prove that the two functors in fact form an equivalence.

\begin{proposition}\label{prop:colim}
  There is a functor $\colim \colon \cat{PBoolD} \to \cat{PBool}$ defined as follows.
  On objects $F \colon L \to \cat{Bool}$, let $\colim(F)$ be the colimit $p_x \colon F(x) \to \coprod F(x) / \sim$.
  On morphisms $(\varphi,\eta) \colon F \to F'$, let $\colim(\varphi,\eta)$ be the morphism $\colim(F) \to \colim(F')$ induced by the cocone $p'_{\varphi(x)} \circ \eta_x \colon F(x) \to \colim(F')$.
  \qed
\end{proposition}

\begin{theorem}\label{thm:equivalence}
  The functors $\PBoolD$ and $\colim$ form an equivalence between the category of piecewise Boolean algebras and the category of piecewise Boolean diagrams.
\end{theorem}
\begin{proof}
  If $P \in \cat{PBool}$, then $\colim(\PBoolD(P)) \cong P$ by Theorem~\ref{thm:colim}. 
  The isomorphism $P \cong \colim(\PBoolD(P))$ is given by $b \mapsto [b]_\sim$.
  If $f \colon P \to P'$, unrolling definitions shows that $\colim(\PBoolD(f))$ sends $[b]_\sim $ to $[f(b)]_\sim$. Therefore $\colim \circ \PBoolD$ is naturally isomorphic to the identity.

  For a 
  diagram $F\colon L \to \cat{Bool}$, fix $P=\colim(F)$.
  Set $\varphi \colon L \to \Sub(P)$ by $x \mapsto p_x[F(x)]$, and
  $\eta_x = p_x \colon F(x) \to p_x[F(x)]$.
  This is a well-defined isomorphism $(\varphi,\eta) \colon F \to \PBoolD(\colim(F))$ by Lemma~\ref{lem:colimit}.
  If $(\psi,\varepsilon) \colon F \to F'$, then $(\psi',\varepsilon')= \PBoolD(\colim(\psi,\varepsilon))$ consists of $\psi' \colon \Sub(\colim(F)) \to \Sub(\colim(F'))$ given by $\psi'(B) = [ \bigcup_{b \in B \cap F(x)} \varepsilon_x(b) ]_\sim$, and $\varepsilon'_B \colon B \to [\varepsilon[B]]_\sim$ given by $\varepsilon'_B(b) = [\varepsilon_x(b)]_\sim$ when $b \in F(x)$. 
  It follows that
  \begin{align*}
    \psi' \circ \varphi(x) 
    = [\varepsilon_x[F(x)]]_\sim
    = \varphi' \circ \psi (x), \\
    (\eta' \psi \cdot \varepsilon)_x (b)
    = [ \varepsilon_x(b) ]_\sim
    = (\varepsilon' \varphi \cdot \eta)_x(b),
  \end{align*}
  whence $(\varphi',\eta') \circ (\psi,\varepsilon) = (\psi', \varepsilon') \circ (\varphi, \eta)$, and $\PBoolD \circ \colim$ is naturally isomorphic to the identity.
\end{proof}

\section{Orientation}\label{sec:orientation}

We have lifted the functor $\Sub$, that is full nor faithful, to an equivalence.
\[\begin{tikzpicture}[xscale=4,yscale=1.75]
  \node (tl) at (0,1) {$\cat{PBool}$};
  \node (tr) at (1,1) {$\cat{PBoolD}$};
  \node (br) at (1,0) {$\cat{Poset}$};
  \draw[->] (tl.south) to node[below] {$\Sub$} (br);
  \draw[->, transform canvas={yshift=1ex}] (tl) to node[above] {$\PBoolD$} (tr);
  \draw[->, transform canvas={yshift=-1ex}] (tr) to node[below] {$\colim$} (tl);
  \node at (.5,1) {$\simeq$};
  \draw[->, transform canvas={xshift=1ex}] (tr) to (br);
  \draw[->, dashed, transform canvas={xshift=-1ex}] (br) to (tr);
\end{tikzpicture}\]
However, the cost was to add the full structure sheaf to $\Sub(P)$.
In this section we reduce to minimal extra structure on a piecewise Boolean domain instead of the full structure sheaf.
In other words: we want to find a converse to the forgetful functor, dashed in the diagram above.
Lemma~\ref{lem:functor} goes towards such a functor, on the level of objects.
However, notice that its proof required making some arbitrary choices. 
We will now fix these choices to obtain a functor.


\begin{proposition}\label{prop:modularatoms}
  Let $L$ be a piecewise Boolean domain. If $x \in L$ is not an atom or 0, we may fix
  $F(x)$ to be the power set of the set of modular atoms in $\Sub(x)$
  in Lemma~\ref{lem:functor}(a).
\end{proposition}
\begin{proof}
  If $x$ has at least four, it follows from a lattice-theoretic characterisation of partition lattices by Sachs~\cite[Theorem~14]{sachs:partitionlattices} that $\Sub(x)$ is dually isomorphic to the lattice of partitions of $\{ \mbox{modular coatoms in } \Sub(x)\op \}$.

  For $x$ of height two or three we may explicitly compute which coatoms of $\Pi_n$ are modular.
  Notice that the element $y=12/34$ is not modular in $\Pi_4$ (see Figure~\ref{fig:smallpartitionlattices}); taking $x=13/2/4$ and $z=13/24$ gives $x \vee (y \wedge z) = x \neq z = (x \vee y) \wedge z$. Similarly, $13/24$ and $14/23$ are not modular. But $123/4$, $124/3$, $134/2$, $234/1$ are modular elements. Hence $\Pi_4$ has 4 modular coatoms.
  Similarly, one can check that all 3 coatoms in $\Pi_3$ are modular.
\end{proof}


\begin{definition}\label{def:orientation}
  An \emph{orientation} of a piecewise Boolean domain $L$ consists of a pointed four-element Boolean algebra $b_a \in F(a)$ for each atom $a \in L$.
%
  A \emph{morphism} 
  \emph{of oriented piecewise Boolean domains} consists of a monotone function $\varphi \colon L \to L'$ satisfying
  \begin{itemize}
    \item if $a \in L$ is an atom, then either $\varphi(a)$ is an atom or $\varphi(a)=0$,
    \item if $a$ is a modular atom in $\Sub(x)$, then $\varphi(a)$ is modular in $\Sub(\varphi(x))$,
  \end{itemize}
  and a map $\eta_a \colon F(a) \to F'(\varphi(a))$
  satisfying $\eta_a(b_a) = b'_{\varphi(a)}$
  for atoms $a \in L$ for which $\varphi(a)$ is a nonmaximal atom.
  The resulting category is denoted $\cat{OPBoolD}$.
\end{definition}

\begin{proposition}
  The functor $\Sub \colon \cat{PBool} \to \cat{Poset}$ extends to orientations as follows. 
  On objects, the orientation is given by $F(B)=B$. The point $b_B$ is the unique element of $\At(C) \cap B$ for an atom $B$ covered by $C$, and $0$ if $B$ is maximal. 
  A morphism $\varphi=\Sub(f)$ extends to orientations by $\eta_B = f \restriction_B \colon B \to f[B]$.
\end{proposition}
\begin{proof}
  First of all, notice that this is well-defined on objects. If $B \in \At(\Sub(P))$ is covered by $C \in \Sub(P)$, say $B=\{0,x,\neg x, 1\}$ for $x \in P$, then precisely one of $x$ and $\neg x$ must be an atom in $C$ (and the other one a coatom). Also, this does not depend on $C$.

  We have to show it is also well-defined on morphisms $f \colon P \to P'$. 
  If $B$ is an atom, say $B=\{0,x,\neg x, 1\}$, then $\varphi(B)=f[B]=\{0,f(x),\neg f(x),1\}$ is clearly either an atom or $\{0,1\}$.
  If $f[B]$ is a nonmaximal atom, then $f[C]$ covers $f[B]$ for some $C \in \Sub(P)$ covering $B$, so $f(b_B)=b'_{f[B]}$ by construction.
  Now suppose $B$ is modular in $\Sub(D)$. Let $A' \subseteq C' \in \Sub(f[D])$; then $A'=f[A]$ and $C'=f[C]$ for some $A,C \in \Sub(D)$, namely $A=f^{-1}(A') \cap D$. 
  Since $A \vee C$ is generated by $A \cup C$, we have $f[A \vee C] = f[A] \vee f[C]$ by~\cite[Proposition~2.4.4]{koppelberg:booleanalgebras}. 
  We may assume $B \cap C=\{0,1\}$, for if $B \subseteq C$ then $f[C] \subseteq f[A] \vee f[B] = f[A \vee B]$ and $f[B]$ is modular in $\Sub(f[D])$.
  Of course always $f[B \cap C] \subseteq f[B] \cap f[C]$. Hence $A' \vee (f[B] \cap C') = f[A \vee B] \cap f[C] \supseteq f[(A \vee B) \cap C] = f[A \vee (B \cap C)] = f[A] \vee f[B \cap C] = f[A]$. Because $A \subseteq C$, the reverse inclusion also holds, and $f[B]$ is modular in $\Sub(f[D])$.

  Finally, this extension is clearly functorial.
\end{proof}

It follows that the forgetful functor $\cat{PBoolD} \to \cat{Poset}$ also extends to orientations as a functor $\cat{PBoolD} \to \cat{OPBoolD}$.

\begin{lemma}\label{lem:orienteddiagram}
  An oriented piecewise Boolean domain $(L,F,b)$ extends uniquely to a piecewise Boolean diagram $F \colon L \to \cat{Bool}$ where $F(a \leq x)(b_a)$ is an atom if $x$ covers an atom $a \in L$.
\end{lemma}
\begin{proof}
  It suffices to show that the requirement in the statement fixes the choice of maps $F(a \leq y)$ for atoms $a \in L$ in Lemma~\ref{lem:functor}(b) in a well-defined way. Pick any $y$ covering $a$, and fix $F(a<y)$ to be the map that sends $b_a$ to an atom in $F(y)$. By diagram~\eqref{eq:orientationchoice}, then $F(a<y')$ maps $b_a$ to an atom for any $y'>a$ for which $z=y \vee y'$ exists (because Theorem~\ref{thm:equivalence} lets us assume that $F=\PBoolD(P)$ for some piecewise Boolean algebra $P$). Hence $F(a<y)$ does not depend on the choice of $y$.
\end{proof}

\begin{lemma}\label{lem:orientedmorphism}
  A morphism of oriented piecewise Boolean domains extends uniquely to a morphism of piecewise Boolean diagrams.
\end{lemma}
\begin{proof}
  We have to extend a map $\eta_a \colon F(a) \to F'(\varphi(a))$, that is only defined on atoms $a \in L$, to a natural transformation $\eta_x \colon F(x) \to F'(\varphi(x))$.
  Let $x \in L$ be nonzero, and let $b' \in F(x)$. Then there is an atom $a \leq x$ and an element $b \in F(a)$ such that $F(a \leq x)(b)=b'$. Define $\eta_x(b) = F'(\varphi(a) \leq \varphi(x))(b')$. Because $a$ and $b$ are unique unless $b' \in \{0,1\}$, this is a well-defined function. Moreover, it is natural by construction. Therefore it is also automatically unique.

  We have to show $\eta_x$ is a homomorphism of Boolean algebras. It clearly preserves $0$ and $\neg$, so it suffices to show that it preserves $\wedge$. Let $b \neq b' \in F(x)$, say $b \in F(a)$ and $b' \in F(a')$ for atoms $a,a' \leq x$. By naturality, we may assume that $x=a \vee a'$. Hence $x$ and $\varphi(x)$ have height 2, and $F(x)$ and $F'(\varphi(x))$ have 8 elements.
  We can now distinguish four cases, depending on whether $b=b_a$ and $b'=b_{a'}$ or not.
  In each case it is easy to see that $\eta_x(b \wedge b') = \eta_x(b) \wedge \eta_x(b')$.
  For example, if $b=b_a$ and $b'=b_{a'}$, then they are distinct atoms in $F(x)$, so $b \wedge b'=0$. But $\eta_x(b)=b'_a$ and $\eta_x(b')=b'_{a'}$ are distinct atoms in $F(\varphi(x))$, so $\eta_x(b) \wedge \eta_x(b') = 0$, too. 
\end{proof}

It follows that morphisms of oriented piecewise Boolean domains preserve directed suprema.

\begin{theorem}
  There is a functor $\cat{OPBoolD} \to \cat{PBoolD}$ that, together with the forgetful functor, forms an isomorphism of categories.
\end{theorem}
\begin{proof}
  Lemmas~\ref{lem:orienteddiagram} and~\ref{lem:orientedmorphism} define the functor on objects and morphisms, respectively; it is functorial by construction.
  Extending an oriented piecewise Boolean domain to a piecewise Boolean diagram and then restricting again to an oriented piecewise Boolean domain leads back to the original.
  Conversely, starting with a piecewise Boolean diagram, restricting it to an oriented piecewise Boolean domain, and then extending, leads back to the original piecewise Boolean diagram by unicity.
  Hence this is an isomorphism of categories.
\end{proof}


\section{Future work}\label{sec:future}

We conclude by listing several directions for future research.

\begin{itemize}
  \item Many examples of piecewise Boolean algebras come from \emph{orthomodular lattices}~\cite{kalmbach:orthomodularlattices,vdbergheunen:colim}. These are precisely the piecewise Boolean algebras that are \emph{transitive} and \emph{joined}: the union $\leq$ of the orders on each commeasurable subalgebra is a transitive relation, and every two elements have a least upper bound with respect to $\leq$~\cite[1.4.22]{kalmbach:orthomodularlattices}; see also~\cite{finch:structure,gudder:partial}. An isomorphism of piecewise Boolean algebras between orthomodular lattices is in fact an isomorphism of orthomodular lattices.\footnote{This was observed in Sarah Cannon's MSc thesis~\cite{cannon:spectrum}, which prompted this work.}  Reformulating these properties in terms of piecewise Boolean domains would  extend our results to orthomodular lattices.

  \item The introduction discussed the analogy between piecewise Boolean diagrams on a piecewise Boolean domains and structure sheaves on a Zariski spectrum. The latter form a topos and hence come with an internal logic~\cite{heunenlandsmanspitters:topos}. However, piecewise Boolean domains are not (pointless) topological spaces. Can we formalise a notion of ``skew sheaf'' over piecewise Boolean domains so that it still makes sense to perform logic in the resulting ``skew topos''?

  \item An obvious question is whether our results extend to piecewise \emph{complete} Boolean algebras.

  \item Although there are many characterisations of finite partition lattices, there is no known equivalence between the category of finite partition lattices and the category of finite sets. For concreteness' sake, it would be very satisfying to explicate the maps $\varphi_{x,y}$ in Lemma~\ref{lem:functor}.

  \item Any C*-algebra $A$ gives rise to a piecewise Boolean algebra $P$. In fact, $\Sub(P)$ determines $A$ up to isomorphism of Jordan algebras~\cite{doeringharding:jordan,hamhalter:pseudojordan,heunen:cccc}. Can our results be used to give an equivalent description of Jordan C*-algebras?
\end{itemize}

\appendix

\bibliographystyle{plain}
\bibliography{piecewisedomains}

\end{document}